\documentclass[11p,reqno]{amsart}

\usepackage[T1]{fontenc}
\usepackage{graphicx}
\usepackage{amssymb,amsthm,amsmath,mathrsfs,bm,braket,marginnote}
\usepackage{enumerate}
\usepackage{appendix}
\usepackage[colorlinks=true, pdfstartview=FitV, linkcolor=blue, citecolor=blue, urlcolor=blue]{hyperref}
\usepackage{multirow}

\usepackage{pgf}
\usepackage{pgfplots}
\usepackage{tikz}
\usetikzlibrary{arrows,calc}
\usepackage{verbatim}
\usetikzlibrary{decorations.pathreplacing,decorations.pathmorphing}
\usepackage[numbers,sort&compress]{natbib}
\usepackage{dsfont}
\numberwithin{equation}{section}
\newtheorem{theorem}{Theorem}[section]

\newtheorem{corollary}[theorem]{Corollary}

\newtheorem{proposition}[theorem]{Proposition}

\theoremstyle{remark}
\newtheorem{remark}[theorem]{Remark}

 \reversemarginpar

\begin{document}

\title[Dualities  for the classical  $\beta$-ensembles]{High-low temperature dualities  for the classical  $\beta$-ensembles}

\author{Peter J. Forrester}
\address{School of Mathematical and Statistics, ARC Centre of Excellence for Mathematical and Statistical Frontiers, The University of Melbourne, Victoria 3010, Australia}
\email{pjforr@unimelb.edu.au}

%\subjclass[2010]{15B52, 15A15, 33E20}
\date{}

\dedicatory{}

\keywords{}

\begin{abstract}
The loop equations for the $\beta$-ensembles are conventionally solved in terms of a $1/N$ expansion.
We observe that it is also possible to fix $N$ and expand in inverse powers of $\beta$. At leading order, for the one-point
function $W_1(x)$ corresponding to the average of the linear statistic $A = \sum_{j=1}^N 1/(x - \lambda_j)$, and specialising the classical weights, this reclaims
well known results of Stieltjes relating the zeros of the classical polynomials to the minimum energy configuration of
certain log-gas potential energies. Moreover, it is observed that the differential equations satisfied by $W_1(x)$ in
the case of classical weights --- which are particular Riccati equations --- are simply related to the  differential equations satisfied by $W_1(x)$ 
in the high temperature scaled limit $\beta = 2\alpha/N$ ($\alpha$ fixed, $N \to \infty$), implying a certain
high-low temperature duality. A generalisation of this duality, valid without any limiting procedure,
 is shown to hold for $W_1(x)$ and all its higher point analogues in the classical  $\beta$-ensembles.
\end{abstract}

\maketitle

\section{Introduction}
The eigenvalue probability density function (PDF), supported on $\mathbb R^N$ and
proportional to
\begin{equation}\label{0.1}
\prod_{l=1}^N e^{- \beta \lambda_l^2/ 2} \,\prod_{1\leq j<k\leq N}|\lambda_k-\lambda_j|^{\beta}
\end{equation}
is said to specify the Gaussian $\beta$-ensemble.
Best known are the cases $\beta = 1$ and $\beta = 2$. Then with $X$ a real 
(respectively complex) Gaussian random matrix, forming the real symmetric
(respectively complex Hermitian) random matrix $H = {1 \over 2} (X + X^\dagger)$
specifies ensembles of random matrices with eigenvalue PDF given by (\ref{0.1}).
For general $\beta > 0$, an ensemble of 
tridiagonal random matrices, with entries on
   and above the diagonal each independent,
 can be specified which realise the eigenvalue PDF (\ref{0.1}) \cite{DE02}.
 
 The PDF (\ref{0.1}) can be written in Boltzmann factor form
 \begin{equation}\label{0.2}
 e^{-\beta U^{\rm G}}, \qquad U^{\rm G} := \sum_{l=1}^N x_l^2/2 - \sum_{1 \le j < k \le N} \log | x_k - x_j|.
 \end{equation}
 It is a classical result due to Stieltjes that the minimum of the potential energy $U^{\rm G}$ is unique and
 occurs at the ordered zeros, $z_1 < z_2 < \cdots < z_N$ say, of the Hermite 
 polynomials $H_N(x)$ \cite{St85a,St85b}; for a recent work on this general theme see \cite{JS20}.
 Thus, for $N$ fixed, the eigenvalues specified by (\ref{0.1}) crystallise to these zeros. Define the Hessian
 matrix $\mathcal H$ at the critical point $\mathbf x = \mathbf z$, where $\mathbf z = (z_1, z_2, \dots, z_N)$  according to
 \begin{equation}\label{0.3}
 \mathcal H = \bigg [ {\partial^2 U^{\rm G} \over \partial x_j \partial x_k} \bigg ]_{j,k=1}^N \Big |_{\mathbf x = \mathbf z}.
 \end{equation}
 Set $y_j = x_j - z_j$. Then to leading order as $\beta \to \infty$, upon a multi-variable Taylor expansion, we see that
 (\ref{0.2}) is proportional to  the Boltzmann factor of an oscillator system specified by
  \begin{equation}\label{0.4}
  \exp \Big ( - {\beta \over 2} \mathbf y \mathcal H \mathbf y^T \Big ).
 \end{equation}
 
 Generally, for a continuous classical statistical mechanical system with a (unique) ground state, expanding
 about this point in configuration space is referred to as the harmonic approximation.
 For the log-gas on a circle of radius $L$, where the potential energy term in
 (\ref{0.2}) is now given by
 \begin{equation}\label{0.5} 
 U^{\rm C} := - \sum_{1 \le j < k \le N} \log | e^{2 \pi i x_k/ L}  - e^{2 \pi i x_j/ L} |, \qquad x_l \in [0, L),
 \end{equation}
 and the   ground state is equal spacing in units of $2 \pi L/N$, a study of the harmonic approximation  has
 been given in \cite{Fo93}, \cite[\S 5]{FJM00}, \cite[\S 4.8.2]{Fo10}
 and \cite{SX21}. Note that for this system, the
 ground state only becomes unique upon fixing one of the particles (say setting $x_1 = 0$) due to the rotation
 invariance. Specifically for (\ref{0.2}), the harmonic approximation was first studied in
 \cite{AJP98}, and subsequently in
 \cite{DE05a} where the analysis proceeded through the tridiagonal realisation of the Gaussian
 $\beta$-ensemble introduced in \cite{DE02}. Much earlier, for the two-dimensional version of
 (\ref{0.2}) --- the so called two-dimensional one-component plasma in a disk --- the harmonic approximation
 was studied in \cite{AJ81}.
 
 The equilibrium statistical mechanical system (\ref{0.2}) admits generalisations, which
 furthermore permit analysis in the $\beta \to \infty$ limit. One generalisation is to consider
 the evolution under Brownian dynamics; see e.g.~\cite[Ch.~11]{Fo10}. Works considering the
 behaviour for $\beta \to \infty$, of this precise model, and variants for other log-potential
 systems in one-dimension, include \cite{AKM12,AV19a,AV19b,Vo19,AHV20,VW21,Tr21}. Another generalisation
 is to consider minor processes associated with the tridiagonal realisation of the $\beta$-ensembles,
 with the Gaussian case (\ref{0.2}) being the simplest \cite{FR02b}.
 Such a process in the limit $\beta \to \infty$ has been studied in the works \cite{BG15,GM20,GK20};
 see also \cite[Remark 13]{VP20}.
 In both these lines of work, a mathematical structure to have shown itself is that of the associated
classical  orthogonal polynomials. Here one recalls (see e.g.~\cite{AW84}) that for a family
of orthogonal polynomials $\{p_n(x) \}$ with three term recurrence
 \begin{equation}\label{0.6} 
 p_{n+1}(x) = (A_n x + B_n) p_n(x) - C_n p_{n-1}(x),
 \end{equation}
 the associated orthogonal polynomials $\{p_n(x;\alpha) \}$ satisfy the modified
 three term recurrence
 \begin{equation}\label{0.6c} 
 p_{n+1}(x;c) = (A_{n+c} x + B_{n+c}) p_n(x;c) - C_{n+c} p_{n-1}(x;c).
 \end{equation}
 
 It is also the case that the weight function for the orthogonality relation of the associated
 classical orthogonal polynomials have appeared in studies of the respective $\beta$-ensembles
 in the scaled high temperature regime, $\beta = 2 \alpha/ N$ and $N \to \infty$
 \cite{ABG12,ABMV13,TT20,  Ma20}.   Specifically, with the weight function $e^{-\beta \lambda^2/2}$
 modified to $e^{-\lambda^2/2}$, and with the corresponding PDF then denoted by use
 of the label "G$^*$" and similarly for the density $\rho_{(1)}(x)$ normalised to integrate to unity,
 we have the  limit formula \cite{ABG12}
  \begin{equation}\label{r3+}
  \rho^{\rm G^*}_{(1),0} (x;\alpha) := 
   \lim_{N \to \infty \atop \beta = 2 \alpha/N}
    \rho^{\rm G^*}_{(1)}(x) =
   {1 \over \sqrt{2 \pi} \Gamma(1 + \alpha)}
  {1 \over | D_{-\alpha}(ix)|^2},
 \end{equation}
 where $D_{-\alpha}(z)$ is the so-called parabolic cylinder function. This same
 expression is known to be the weight function for the orthogonality relation of the associated
 Hermite polynomials \cite{AW84}.
 
 In relation to the scaled high temperature regime, it is also known that  \cite{ABG12}
   \begin{equation}\label{u3+}
W_1^{0,\rm G^*}(x;\alpha) :=  \lim_{N \to \infty \atop \beta = 2 \alpha/N}
 \left\langle\sum_{j=1}^N\frac{1}{x-\lambda_j}\right\rangle^{\rm G^*}
=
- {1 \over \alpha} {d \over dx}\Big ( e^{- x^2/4}  \log D_{-\alpha}(ix) \Big );
\end{equation}
in fact it was after first deriving this expression that (\ref{r3+}) was deduced in \cite{ABG12}.
Thus
  \begin{equation}\label{u3++}
  W_1^{0,\rm G^*}(x;\alpha)  = \int_{-\infty}^\infty { \rho^{\rm G^*}_{(1),0} (y;\alpha)  \over x -  y} \, dy,
  \end{equation}
  and so $\rho^{\rm G^*}_{(1),0}$ can be deduced from knowledge of $ W_1^{0,\rm G^*}$ by
  applying the formula for the inverse Stieltjes transform.
Starting from the known differential equation for the  parabolic cylinder function
\cite[\S 12.2]{DLMF}, and with $g(z):= e^{z^2/4} D_{-\alpha}(z)$,
we have that $g(z)$ satisfies the second order linear differential equation
  \begin{equation}\label{u4+}
  g''(z) - z g'(z) - a g(z) = 0.
  \end{equation}
  On the other hand, Stieltjes result gives that in relation to (\ref{0.1}) itself (this will
  be labelled by the use of "G" as used in (\ref{0.2})), now taking the low temperature limit with $N$ fixed,
  we have
    \begin{equation}\label{u5+}
 W_1^{0,\rm G}(x;N) :=  \lim_{\beta \to \infty}
 \left\langle\sum_{j=1}^N\frac{1}{x-\lambda_j}\right\rangle^{\rm G} = {d \over d x} \log H_N(x).
    \end{equation}  
    The differential equation satisfied by $H_N(x)$ is
  \begin{equation}\label{u6+}   
f''(x) - 2 x f'(x) + 2N f(x) = 0.
 \end{equation}  
 Thus (Corollary \ref{c2.3})  
 \begin{equation}\label{w1}
   {i \over \sqrt{2} N} W_1^{0, \rm G}(ix/\sqrt{2};N) 
   = W_1^{0, \rm G^*}(x; - N).
   \end{equation}
   
   It is shown in Remark \ref{r2.4}, point 1, that the known functional equation for the
   moments of the density of (\ref{0.1}), for $N$ and $\kappa:=\beta/2$ fixed
   \cite{DE05}
  \begin{equation}\label{w4+} 
  m_{2p}^{\rm G}(N,\kappa) =  (-1)^{p+1} \kappa^{-p-1}     m_{2p}^{\rm G}(-\kappa N,\kappa^{-1} ),
  \end{equation}
  permits an independent derivation of (\ref{w1}). With regards to (\ref{w4+}), note that
  as well as mapping $\kappa$ to $1/\kappa$, and thus high temperature to low temperature, 
  the variable $N$ is mapped to $-\kappa N$, which can thus no longer be interpreted as the
  number of eigenvalues in (\ref{0.1}). Nevertheless, it is well defined since $  m_{2p}^{\rm G}$
  is a polynomial of degree $p+1$ in $N$; see e.g.~\cite{WF14}. 
 
 In addition to providing the details of the derivation of (\ref{w1}) and its relationship to (\ref{w4+}),
 we will extend the working to the Laguerre
 and Jacobi $\beta$-ensembles. Thus we will proceed
 by considering the differential equation satisfied by $g(x)$
 in the formula 
 $ W_1^{0,\rm C^*}(x) \propto {d \over dx} \log g(x)$, where $\rm C^*$ denotes the scaled high temperature limit of the
 respective classical $\beta$-ensemble, and that satisfied by $f(x)$ in the formula
 $ W_1^{0,\rm C}(x) = {d \over dx}  \log f(x)$ corresponding to the low temperature limit $\beta \to \infty$ as implied
 by the results of Stieltjes relating $f(x)$ to a classical orthogonal polynomial.
 Moreover, we will  show that the functional equation (\ref{w1}) can be generalised to higher point
 extensions of the averages in (\ref{u3+}) and (\ref{u5+}), specified for $n=2,3$ by appropriate limits of
\begin{equation}\label{1.4c+} 
 \overline{W}_n(x_1,\dots,x_n; \kappa, N)  :=
\Big  \langle    \prod_{l=1}^n (A_l  - \langle A_l \rangle )
 \Big \rangle
, \quad A_i =  \sum_{j=1}^N\frac{1}{x_i-\lambda_j}.
\end{equation}
For general $n \ge 2$ define $\bar{W}_n$ by the requirement that it be the generating function for the
mixed cumulants $\{ \mu_{p_1,\dots,p_n}\}$ as indexed by $p_1,\dots,p_n$,
\begin{equation}\label{1.16a}
\overline{W}_n(x_1,\dots,x_n; N, \kappa) =
 {1 \over x_1 \cdots x_n} \sum_{p_1,\dots,p_n = 0}^\infty {\mu_{p_1,\dots,p_n}(N,\kappa) \over
 x_1^{p_1} \cdots x_n^{p_n}},
 \end{equation}
 considered as a formal series. With $n = n_0 \ge 4$ this is can equivalently be given as a multinomial in
 the averages on the RHS of (1.16) with $n=2,\dots,n_0$, homogeneous of degree $n$ in the total number of
 factors of the form $(A_l - \langle A_l \rangle )$. The resummation holds as an analytic function when
 the eigenvalue support is compact (Jacobi case) and the $x_i$ are large enough; otherwise (\ref{1.16a})
 is the asymptotic expansion of the definition in terms of averages.
 The definition in terms of averages and moments in the case $n = 1$ is
 \begin{equation}\label{1.4d}
 \overline{W}_1(x_1;N,\kappa) = \langle A_1 \rangle = {1 \over x} \sum_{p=0}^\infty {m_p(N,\kappa) \over x^p};
  \end{equation}
in relation to the first equality note that (\ref{1.4c+}) is not appropriate as it vanishes for $n=1$, while in relation to the
second note that for $n=1$ the cumulants and the moments are the same quantity.

\begin{proposition}\label{p1.1}
For the Gaussian $\beta$-ensemble with weight
$e^{-\kappa x^2}$, denote the quantity $\overline{W}_n$ specified in the above paragraph by
$\overline{W}_n^{\rm G}(x_1,\dots,x_n;N,\kappa)$. Specifically the expansion (\ref{1.16a}) is to be used to provide
a meaning for continuous values of $N$.
We have
\begin{equation}\label{1.4e} 
 \overline{W}_n^{\rm G}(x_1,\dots,x_n; N, \kappa)
 =  ( - i /   \sqrt{\kappa })^n    \overline{W}_n^{\rm G}(i  \sqrt{\kappa} x_1,\dots, i  \sqrt{\kappa}  x_n ; - \kappa N, \kappa^{-1}).
\end{equation}
\end{proposition}

To see the relevance of (\ref{1.4e}) to (\ref{w4+}),  substitute the moment expansion (\ref{1.4d})
and equate coefficients of like powers of $1/x$ on both sides to reclaim the latter.
 The proof of Proposition \ref{p1.1} is given in Section \ref{S2}, along with other  high-low temperature
  dualities in the Gaussian case such as (\ref{w1}). Analogous workings are undertaken in Sections \ref{S3} and \ref{S4} for the
  Laguerre and Jacobi $\beta$-ensembles respectively.

 \section{High-low temperature duality for $W_1^{\rm G}$ and generalisation to $W_n^{\rm G}$}\label{S2}
 \subsection{The zero temperature limit of $W_1^{\rm G}$}
 Generalising (\ref{0.1}), consider the family of eigenvalue PDFs proportional to
 \begin{equation}\label{1.2}
\prod_{l=1}^Nw(\lambda_l)\,\prod_{1\leq j<k\leq N}|\lambda_k-\lambda_j|^{\beta},
\end{equation}
where $w(\lambda)$ is referred to as the weight function. We use the notation
ME${}_{\beta,N}[w]$ in reference to (\ref{1.2}), and the terminology 
$\beta$-ensemble, with the name associated with the weight specified as an adjective.
Thus, for example, ME${}_{\beta,N}[e^{-\beta \lambda^2/2}]$ is referred to as the Gaussian $\beta$-ensemble,
as used in association with (\ref{0.1}).

 Let $A = \sum_{j=1}^N a(\lambda_j)$ be a general linear statistic, and consider its average
 $\langle A \rangle_{ {\rm ME}_{\beta,N}[w]}$. Let
 $\rho_{(1)}(\lambda)$ denote the corresponding eigenvalue density.
 This is specified by the requirement that $\int_a^b\rho_{(1)}(\lambda) \, \mathrm{d}\lambda$ 
 be equal to the expected number of eigenvalues in a general
interval $[a,b]$. It relates to the average of the linear statistic $A$ by
 \begin{equation}\label{1.3}
 \langle A \rangle_{ {\rm ME}_{\beta,N}[w]} = \int_{-\infty}^\infty a(\lambda) \rho_{(1)}(\lambda) \, d \lambda.
\end{equation}

The two-point correlation $\rho_{(2)}(\lambda_1, \lambda_2)$ can be specified by the
requirement when divided by
$\rho_{(1)}(\lambda_2)$, it is equal to the
eigenvalue density at $\lambda_1$, given there is an eigenvalue at $\lambda_2$. It relates to the
average of the product of two linear statistics
\begin{equation}\label{1.4}
 \langle A B \rangle_{ {\rm ME}_{\beta,N}[w]} =  \int_{-\infty}^\infty 
  \int_{-\infty}^\infty
 a(\lambda_1) b(\lambda_2)    \Big (
  \rho_{(2)}(\lambda_1, \lambda_2)  + \delta(\lambda_1 - \lambda_2)   \rho_{(1)}(\lambda) \Big )
  \, d \lambda_1  d \lambda_2.
  \end{equation}
 Equivalently, for the covariance
 \begin{align}\label{1.4a}
 {\rm Cov} \, (A,B) &:=     \langle (A  - \langle A \rangle )
 (B - \langle B) \rangle_{{\rm ME}_{\beta,N}[w]}
  \\ \nonumber & =
 \int_{-\infty}^\infty 
  \int_{-\infty}^\infty
 a(\lambda_1) b(\lambda_2)    \Big (
  \rho_{(2)}^T(\lambda_1, \lambda_2)  + \delta(\lambda_1 - \lambda_2)   \rho_{(1)}(\lambda) \Big )
  \, d \lambda_1  d \lambda_2,
  \end{align}
  where $   \rho_{(2)}^T(\lambda_1, \lambda_2)  :=   \rho_{(2)}(\lambda_1, \lambda_2)  - \rho_{(1)}(\lambda_1) 
  \rho_{(1)}(\lambda_2)$.
  
  Introduce the linear statistics and covariances
   \begin{equation}\label{1.4b} 
   \overline{W}_1(x) := \left\langle\sum_{j=1}^N\frac{1}{x-\lambda_j}\right\rangle_{{\rm ME}_{\beta,N}[w]}, \quad
   \overline{W}_2(x_1,x_2)  :=
    {\rm Cov} \, \bigg (   \sum_{j=1}^N\frac{1}{x_1-\lambda_j}, \sum_{j=1}^N\frac{1}{x_2-\lambda_j} \bigg )_{{\rm ME}_{\beta,N}[w]}.
\end{equation}
 Later, we will have use too for the $n$-point generating function (\ref{1.16a}) viewed in terms of the
 averages (\ref{1.4c+}), which extends the above definition of
 $\overline{W}_2(x_1,x_2)$.
For the $\beta$-ensemble with Gaussian weight $w(x) = e^{- \beta x^2/ 2}$, the quantities (\ref{1.4b}) are
inter-related by the particular loop equation
\cite{BEMN10}, \cite{BMS11}, \cite{MMPS12}, \cite{WF14}
\begin{equation}\label{L1}
\Big ( ( \kappa - 1)  {\partial \over \partial x_1} - 2 \kappa  x_1 \Big ) \overline{W}_1^{\rm G}(x_1;N,\kappa) + 2 N \kappa
+ \kappa \Big ( \overline{W}_2^{\rm G}(x_1, x_1;N,\kappa) + (\overline{W}_1^{\rm G}(x_1;N,\kappa) )^2 \Big ) = 0,
\end{equation}
where $\kappa = \beta/2$ as in (\ref{w4+}).
%and
%\begin{multline}\label{L2}
%\Big ( ( \kappa - 1) {\partial \over \partial x_1} - 2 \kappa  x_1 \Big ) \overline{W}_2(x_1, x_2) +
%{\partial \over \partial x_2} \Big ( {\overline{W}_1(x_1) - \overline{W}_1(x_2) \over x_1 - x_2 } \Big ) \\
%+ \kappa \Big ( \overline{W}_3(x_1,x_1,x_2) + 2 \overline{W}_2(x_1, x_2) \overline{W}_1(x_1) \Big ) = 0,
%\end{multline}
The equation (\ref{L1}) is in fact the first in an infinite hierarchy
involving $\{   \overline{W}_n  \}_{n=1}^\infty $; see subsection \ref{S2.3} below.

We see that the equation (\ref{L1}) contains two unknowns; generally the $n$-th loop
equation involves $\{   \overline{W}_n(x_1,\dots,x_k;N,\kappa)  \}_{k=1}^{n+1}$. It has been known for some
time (see \cite{Mi83}) that a triangular system of equations results from an appropriate $1/N$ expansion
of each of the $ \overline{W}_n$, with an essential point being that this quantity decays at leading order as $N^{2-n}$.
In the recent work \cite{FM21} it was shown that a triangular system also results in the case
$\beta$ is proportional to $1/N$, even though each $ \overline{W}_n$ is proportional to $N$.
Here our interest is $N$ fixed, but $\beta$ large. For this we make use of the expansions
\begin{equation}\label{L3}
 \overline{W}_1^{\rm G}(x;N,\kappa)  = W_1^{0, \rm G}(x;N) + {1 \over \kappa} W_1^{1, \rm G}(x;N) + \cdots, \quad
  \overline{W}_2^{\rm G}(x_1, x_2;N,\kappa)  = {1 \over \kappa} W_2^{0,G}(x_1, x_2;N) +  \cdots,
    \end{equation}
    where higher order terms are in higher powers of $1/\kappa$.
    Note that the meaning of 
$W_1^{\rm 0,G}(x;N)$ here is consistent with (\ref{u5+}).
  
 In relation to justifying (\ref{L3}), in particular the first relation, as is consistent with
 (\ref{1.4d}) recall 
 that
   the moments of the spectral density $\{ m_{2p}^{\rm G} \}$ relate to 
  the linear statistic $ \overline{W}_1^{\rm G}(x)$ by the large $x$ expansion
\begin{equation}\label{v1}
 \overline{W}_1^{\rm G}(x;N,\kappa) = {1 \over x} \sum_{p=0}^\infty {m_{2p}^{\rm G}(N,\kappa) \over x^{2p} }, \qquad m^{\rm G}_{2p}(N,\kappa)  = \int_{-\infty}^\infty x^{2p}  \rho_{(1)}(x) \, dx,
\end{equation}
where use has been made of the fact that since the Gaussian weight is symmetric about the origin, the odd moments vanish.
Thus the expansion for $\overline{W}_1^{\rm G}(x;N,\kappa) $ in (\ref{L3}) implies the corresponding  large $\beta$ expansion of the moments
\begin{equation}\label{v2}
 m_{2p}^{\rm G}(N,\kappa)  = m_{2p}^{0, \rm G}(N) + {1 \over \kappa} m_{2p}^{1,\rm G}(N) + \cdots,
 \end{equation}
 where the $m_{2p}^{k,\rm G}(N)$ are independent of $\beta$. 
 Rescaling $m_{2p}^{\rm G}(N,\kappa) = 2^{-p} \tilde{m}_{2p}^{\rm G}(N,\kappa)$ for notational convenience,
 for low orders we have the exact
 results \cite[\S 4.3]{WF14}
 \begin{align}\label{2.12}
 \tilde{m}_0^{\rm G}(N,\kappa)  & = N  \nonumber \\
 \tilde{m}_2^{\rm G}(N,\kappa)  & = N^2 - N + \kappa^{-1} N   \nonumber \\
 \tilde{m}_4^{\rm G}(N,\kappa)  & = 2 N^3 - 5N^2 + 3 N + \kappa^{-1} (5 N^2 - 5 N)  + \cdots \nonumber \\
 \tilde{m}_6^{\rm G}(N,\kappa)  & = 5 N^4 - 22 N^3 + 32 N^2 - 15 N  + \kappa^{-1} (22 N^3 - 54 N^2 + 32 N)  + \cdots \nonumber \\
 \tilde{m}_8^{\rm G}(N,\kappa) & = 14 N^5 -93 N^4 + 234 N^3 - 260 N^2 + 105 N + \kappa^{-1} ( 93 N^4 - 398 N^3 + 565 N^2 - 260 N) + \cdots,
 \end{align}
 where terms not shown are higher order in $\kappa^{-1}$ (each $\tilde{m}_{2p}$ is a polynomial in $\kappa^{-1}$ of degree
 $p$). This is consistent with (\ref{v2}) and thus the first expansion in (\ref{L3}). Furthermore, 
with $N$ large the $\kappa$ dependence in (\ref{L3})
 has previously been established \cite[\S 3]{WF14}.
 
 The formal justification for fixed $N$ is to use the Laplace method of asymptotic analysis, whereby the integrand corresponding to the averages in 
 (\ref{v1}) is expanded about the critical point $\boldmath \lambda = \mathbf z$. Let $A_i := \sum_{j=1}^N 1/ (x_i - \lambda_j)$.
 Recalling the working leading to (\ref{0.4}), this shows
 $$
 \langle A_1 \rangle = \sum_{j=1}^N {1 \over x - z_j} + \langle \mathbf y \mathcal H_1 \mathbf y^T \rangle_Q + \cdots, \qquad
 \langle A_2 \rangle =  \langle \mathbf y \mathcal H_2 \mathbf y^T \rangle_Q + \cdots.
 $$
 Here higher order terms involve powers of $y_j^{2p}$ for $p \ge 2$,
 $\mathcal H_1$ and $\mathcal H_2$ are some $N \times N$ symmetric matrices, and the average over $Q$ is with respect to the PDF corresponding to
(\ref{0.4}). Changing variables $\sqrt{\beta} \mathbf y \mapsto \mathbf y$ implies (\ref{L3}). 
 
 Substituting (\ref{L3}) in (\ref{L1}) and equating leading terms for large $\beta$ gives a differential equation
 involving $W_1^{0, \rm G}(x)$ only,
 \begin{equation}\label{2.13}
 \Big ( {1 \over 2} {d \over d x} - x \Big ) W_1^{0, \rm G}(x;N) + N + {1 \over 2} ( W_1^{0, \rm G}(x;N) )^2 = 0.
  \end{equation}
  This Riccati equation can readily be solved in terms of Hermite polynomials.
  
  \begin{proposition}\label{p2.1}
  We have
   \begin{align}\label{2.14}
    W_1^{0, \rm G}(x;N) & = {d \over d x} \log H_N(x) = {N \over x} + {d \over d x} \Big ( 1 + N!  \sum_{m=1}^\infty {(-1)^m (2 x)^{-2m} \over m! (N - 2m)!} \Big ) \nonumber \\
    & = {N \over x} + {N^2 - N \over 2 x^3} + {2N^3  - 5N^2 + 3 N \over 2^2 x^5} + {5N^4  - 22 N^3 + 32 N^2 -15 N \over 2^3 x^7} + \cdots 
   \end{align}
   Furthermore, writing 
   \begin{equation}\label{2.14a}  
     W_1^{0, \rm G}(x;N) = {1 \over x} \sum_{p=0}^\infty {m_{2p}^{0, \rm G}(N) \over x^{2p} }
     \end{equation}
     as is consistent with (\ref{v1}) and (\ref{v2}), we have that $\{  m_{2p}^{0, \rm G} \}$ satisfy the recurrence
   \begin{equation}\label{2.15}  
   2 m_{k+2}^{0, \rm G}(N)  = \sum_{s=0}^{k/2} m_{2s}^{0, \rm G}(N)  m_{k-2s}^{0, \rm G}(N)  - (k + 1) m_{k}^{0, \rm G}(N), \qquad m_0^{0, \rm G}(N) = N,
   \end{equation}
   valid for $k=0,2,4,\dots$.
   \end{proposition}
   
   \begin{proof}
  Setting
   $$
    W_1^{0, \rm G}(x;N) = {d \over d x} \log  f(x), \qquad f(x) \mathop{\sim}\limits_{x \to \infty} c x^N,
   $$
   where $c$ is a nonzero constant, we see from (\ref{2.13}) that $f$ satisfies the linear second order differential equation
  \begin{equation}\label{2.16}   
  f''(x) - 2 x f'(x) + 2 N f(x) = 0.
  \end{equation}
  The unique solution of this equation satisfying the required boundary condition is the Hermite polynomial $H_N(x)$.
  The recurrence (\ref{2.15}) follows by substituting (\ref{2.14a}) in (\ref{2.13}) and equating like powers of $1/x$.
  \end{proof}

  \begin{remark}
  1.~We see that the explicit form of the moments $\{  m_{2p}^{0, \rm G}  \}$ as read off from (\ref{2.14}) and (\ref{2.14a})
  agree with the leading (with respect to $1/\kappa$) terms in (\ref{2.12}). \\
  2.~For $\beta \to \infty$, the result of Stieltjes discussed below (\ref{0.2}) tells us that
  $$
   W_1^{0, \rm G}(x;N) =  \sum_{j=1}^N {1 \over x - z_j},
   $$
   where $\{ z_j \}$ denote the zeros of $H_N(x)$. Thus
  \begin{equation}\label{r1}
    W_1^{0, \rm G}(x;N) =  {d \over d x} \log H_N(x)
  \end{equation}
  in agreement with (\ref{2.14}). \\
  3.~It follows from (\ref{2.15}) that each $m_{2p}^{0, \rm G}$ is a polynomial of
  degree $p+1$ in $N$ which vanishes when $N = 0$ and $N = 1$. The coefficient of $N^{p+1}$
  is recognised as $C_p/2^p$, where $\{ C_p \}$ denotes the Catalan numbers. Indeed,
  writing $m_{2p}^{0, \rm G} \mathop{\sim} a_p N^{p+1}/2^p$ for $N \to \infty$ we see that
  (\ref{2.15}) simplifies to read
   \begin{equation}\label{2.17} 
   a_{p+1} = \sum_{s=0}^p a_s a_{p - s}, \qquad a_0 = 1,
   \end{equation}
   which has the unique solution $a_p = C_p$. This is in keeping with the well known result 
   (see e.g.~\cite{KM16}) that the density of the
   zeros of the Hermite polynomials, scaled by $\sqrt{2/N}$, is given by the Wigner
   semi-circle law
    \begin{equation}\label{2.18a}
    \rho_{(1)}^{\rm W}(x) = {1 \over \pi} \sqrt{ 1 - (x/2)^2}, 
    \end{equation}
    supported on $|x| < 2$, as follows from the moment formula
    \begin{equation}\label{2.18b}  
    \int_{-2}^2 x^{2p}  \rho_{(1)}^{\rm W}(x)  \, dx = C_p.
   \end{equation}  
  \end{remark}
  
  \subsection{High-low temperature duality of $W_1^{\rm G}$}\label{S2.2}
  Making the replacements
  \begin{equation}\label{s1}
  x \mapsto i x/ \sqrt{2}, \qquad W_1^{0, \rm G}(i x/ \sqrt{2};N) \mapsto (\sqrt{2} N/ i) g(x;N)
 \end{equation} 
 in (\ref{2.13}) shows
  \begin{equation}\label{s2}
  \Big ( -  {\partial \over \partial x} - x \Big ) g(x;N)  + 1 - N (  g(x;N))^2 = 0,
  \qquad   g(x;N)) \mathop{\sim}\limits_{x \to \infty} {1 \over x}.
  \end{equation}  
  This same differential equation (up to the meaning of $N$) is known in the study of the
  Gaussian $\beta$-ensemble ME${}_{\beta, N} [ e^{- x^2/2} ]$, now scaled at high temperature
  by setting $\beta = 2 \alpha / N$ then taking $N \to \infty$ \cite{ABG12}.
  Hence writing
  \begin{equation}\label{s3} 
{1 \over N}  \left\langle\sum_{j=1}^N\frac{1}{x-\lambda_j}\right\rangle_{{\rm ME}_{\beta,N}[e^{-x^2/2}]}
\bigg |_{\beta = 2 \alpha / N} = W_1^{0, \rm G^*}(x;\alpha) + {1 \over N} W_1^{1, \rm G^*}(x;\alpha) + \cdots
  \end{equation}  
  (here we are using G${}^*$ to indicate the Gaussian $\beta$-ensemble based on the weight $e^{- x^2/2}$
  rather than on the weight $e^{-\beta x^2/2}$ as relates to Proposition \ref{p2.1})
  one has that $W_1^{0, \rm G^*}(x;\alpha) $ satisfies (\ref{s2}) but with $-N$ replaced by $\alpha$.
  We can thus relate $W_1^{0, \rm G}$ to $W_1^{0, \rm G^*}$, and similarly relate the corresponding moments.
  
  \begin{corollary}\label{c2.3}
  We have the identity (\ref{w1}).
  Also, with the moments $\{ m_{2p}^{0, \rm G}(N) \}$ specified as in (\ref{2.14a}), and the moments
$\{ m_{2p}^{0, \rm G^*}(\alpha) \}$  specified according to
  \begin{equation}\label{w2}
  W_1^{0, \rm G^*}(x;\alpha) = {1 \over x} \sum_{p=0}^\infty { m_{2p}^{0, \rm G^*}(\alpha) \over x^{2p}},
   \end{equation}
   it follows from (\ref{w1}) that
   \begin{equation}\label{w3}  
   {(-2)^p \over N } m_{2p}^{0, \rm G}(N)  = m_{2p}^{0, \rm G^*}(-N).
   \end{equation}
   \end{corollary}
   
  \begin{remark}\label{r2.4}
  1.~For the Gaussian $\beta$-ensemble with weight $e^{-\beta x^2/2}$  it is known that the moments $m_{2p}^{\rm G}(N,\kappa)$,
  for fixed $N$ and $\kappa$, 
  satisfy the functional equation (\ref{w4+}). 
  From the definitions $m_{2p}^{\rm G}(N,\kappa) =  \beta^{-p} m_{2p}^{\rm G^*}(N,\kappa)$.   
   Setting $\kappa = \alpha/N$, dividing by $N$ and taking $N \to \infty$, it follows
  \begin{equation}\label{w3d}    
   m_{2p}^{\rm G^*}(\alpha) = (-1)^{p+1} (2^p /\alpha) \lim_{N \to \infty}  m_{2p}^{\rm G}(-\alpha,N/\alpha)
   =  (-1)^{p+1} (2^p /\alpha) m_{2p}^{0,\rm G} \Big |_{N \mapsto - \alpha},
  \end{equation}
  in agreement with (\ref{w3}). We remark that in  \cite{DS15} this duality between high and low temperature for the
  Gaussian $\beta$-ensemble has been used to study the density of states in the scaled high temperature limit from
  knowledge of underlying structures at zero temperature. \\
2.~The
   solution of the differential equation (\ref{s2}) with $-N = \alpha$ is \cite{ABG12} (see \cite{AB19} for a 
   generalisation)
   \begin{equation}\label{w5} 
   W_1^{0, \rm G^*}(x;\alpha) = {x \over 2 \alpha} - {1 \over \alpha} {d \over d x} \log D_{-\alpha}(ix),
  \end{equation}
  where $D_{-\alpha}(z)$ is the so-called parabolic cylinder function. Now (see e.g.~\cite{Wu19})
  $$
  D_N(z) = e^{- z^2/4} 2^{-N/2} H_N(z/ \sqrt{2})
  $$
  and so
  $$
  W_1^{0, \rm G^*}(ix;\alpha) \Big |_{\alpha = - N} = {1 \over i N} {d \over dx} \log H_N(x/\sqrt{2}).
  $$
  With knowledge of (\ref{w1}), this reclaims (\ref{r1}).   
    \end{remark}
  
 \subsection{High-low temperature duality of $W_n^{\rm G}$}\label{S2.3}  
Our aims in this subsection are to prove Proposition \ref{p1.1}, and to proceed to
deduce from the proposition generalisations of (\ref{w1}) and (\ref{w4+}).

\smallskip
\noindent {\it Proof of Proposition \ref{p1.1}.} \: \: 
For $n \ge 2$ the $n$-th loop equation for the Gaussian $\beta$-ensemble with weight $e^{-\beta x^2/2}$,
and thus the generalisation of (\ref{L1}) which corresponds to the case $n=1$,
is (see the references listed in relation to (\ref{L1}))
\begin{align}
0&=\left[(\kappa-1)\frac{\partial}{\partial x_1}- 2 \kappa x_1\right]\overline{W}_n(x_1,J_n;N,\kappa)\nonumber
\\&\quad+ \sum_{k=2}^n\frac{\partial}{\partial x_k}\left\{\frac{\overline{W}_{n-1}(x_1,\ldots,\hat{x}_k,\ldots,x_n;N,\kappa)-\overline{W}_{n-1}(J_n;N,\kappa)}{x_1-x_k}\right\}\nonumber
\\&\quad+\kappa\left[\overline{W}_{n+1}(x_1,x_1,J_n;N,\kappa)+\sum_{J\subseteq J_n}\overline{W}_{|J|+1}(x_1,J;N,\kappa)\overline{W}_{n-|J|}(x_1,J_n\setminus J;N,\kappa)\right].\label{3.10}
\end{align}
Here the notation $\hat{x}_k$ indicates that the variable $x_k$ is not present in the argument, and thus $\overline{W}_{n-1}(x_1,\ldots,\hat{x}_k,\ldots,x_n;N,\kappa) =
\overline{W}_{n-1}( \{x_j\}_{j=1}^n \setminus \{x_k \};N,\kappa)$. Also for $n \ge 2$, $J_n=(x_2,\ldots,x_n)$, while $J_1=\emptyset$.
In (\ref{3.10}) change variables $x_j \mapsto i \sqrt{\kappa} x_j$, replace $\kappa$ by $\kappa^{-1}$ and $N$ by $\tilde{N}$ where
$\tilde{N}$ is arbitrary. Then with
$$
\tilde{W}_n^{\rm G}(x_1,\dots,x_j;N,\kappa) = \overline{W}_n^{\rm G}(i \sqrt{\kappa} x_1,\dots, i \sqrt{\kappa} x_j; \tilde{N}, \kappa^{-1})
$$
we see that $\{ \tilde{W}_j^{\rm G} \}$ satisfy 
\begin{align}
0&=i \left[(\kappa-1)\frac{\partial}{\partial x_1}- 2 \kappa x_1\right]\tilde{W}_n(x_1,J_n; \tilde{N},\kappa)\nonumber
\\&\quad-  \kappa^{1/2} \sum_{k=2}^n\frac{\partial}{\partial x_k}\left\{\frac{\hat{W}_{n-1}(x_1,\ldots,\hat{x}_k,\ldots,x_n; \tilde{N},\kappa)-\tilde{W}_{n-1}(J_n;\tilde{N},\kappa)}{x_1-x_k}\right\}\nonumber
\\&\quad+\kappa^{1/2}  \left[\tilde{W}_{n+1}(x_1,x_1,J_n;\tilde{N},\kappa)+\sum_{J\subseteq J_n}\hat{W}_{|J|+1}(x_1,J;\tilde{N},\kappa)\hat{W}_{n-|J|}(x_1,J_n\setminus J; \tilde{N},\kappa)\right].\label{3.10+}
\end{align}
Now substituting
$$
\tilde{W}_n^{\rm G}(x_1,\dots,x_j;\tilde{N},\kappa)  = (i \sqrt{\kappa})^{n}  \hat{W}_n^{\rm G}(x_1,\dots,x_j;\tilde{N},\kappa) 
$$
shows that the equations satisfied by $ \hat{W}_n^{\rm G}$ are precisely (\ref{3.10}) for each $n=2,3,\dots$.

We can check too that with the same change of variables and mappings of the previous paragraph, the
loop equation for $n=1$ (\ref{L1}) is similarly transformed to its original form, provided $\tilde{N} = - \kappa N$.
Thus all the loop equations are invariant under the mapping implied by the right hand side of (\ref{1.4e}).
Since upon a $1/N$ expansion, the validity of which is rigorously established in \cite{BG12},  the loop equations
uniquely determine $\{ W_n^{\rm G} \}$, the functional equation (\ref{1.4e}) has been validated. \hfill $\square$

\smallskip

We can use (\ref{1.4e}) to deduce the analogue of (\ref{w1}), relating the low temperature
$n$-point quantity $W_n^{0, \rm G}$ in the $\beta \to \infty$ expansion
\begin{equation}\label{2.29m}
\overline{W}_n^{\rm G}(x_1,\dots, x_n;N,\kappa) = {1 \over \kappa^{n-1}} W_n^{0, \rm G}(x_1,\dots, x_n;N,\kappa) + \cdots,
\end{equation}
(cf.~(\ref{L3}); the derivation involving the Laplace method of asymptotic expansion is the same, making essential use of the
structure of the $\overline{W}_n^{\rm G}$ when expressed in terms of the averages
on the RHS of (\ref{1.4c+}) as noted below (\ref{1.16a})) to the high temperature $n$-point quantity $W_n^{0, \rm G^*}$ in the $N \to \infty$ expansion
\begin{equation}\label{2.29a}
{1 \over N} \overline{W}_n^{\rm G^*}(x_1,\dots,x_n;N,\kappa)) \Big |_{\kappa =  \alpha / N}  =
{W}_n^{0, \rm G^*}(x_1,\dots,x_n;\alpha) + \cdots
\end{equation}
(cf.~(\ref{s3})).  Noting the analogues of (\ref{1.16a})
\begin{equation}\label{x1}
 {W}_n^{0, \rm G}(x_1,\dots,x_n;N) =
 {1 \over x_1 \cdots x_n} \sum_{p_1,\dots,p_n = 0}^\infty {\mu_{p_1,\dots,p_n}^{0, \rm G}(N) \over
 x_1^{p_1} \cdots x_n^{p_n}}.
 \end{equation}
 and
 \begin{equation}\label{x2}
 {W}_n^{0, \rm G^*}(x_1,\dots,x_n; \alpha) =
 {1 \over x_1 \cdots x_n} \sum_{p_1,\dots,p_n = 0}^\infty {\mu_{p_1,\dots,p_n}^{0, \rm G^*}(\alpha) \over
 x_1^{p_1} \cdots x_n^{p_n}},
 \end{equation}
 this can be equivalently be written as a generalisation of (\ref{w3}).

\begin{corollary}
We have
\begin{equation}\label{w1j}
 -\Big ( {-i \over  \sqrt{2} }  \Big )^n {1 \over N} W_n^{0, \rm G}(ix_1/\sqrt{2},\dots,ix_n/\sqrt{2};N) 
   = W_n ^{0, \rm G^*}(x_1,\dots,x_n; - N).
   \end{equation}
   Equivalently
\begin{equation}\label{w1k}  
{ (- \sqrt{2} )^{p_1 + \cdots p_n} \over N}
\mu^{0, \rm G}(p_1,\dots,p_n;N) =
\mu^{0, \rm G^*}(p_1,\dots,p_n;-N).
 \end{equation}
\end{corollary}

\begin{proof}
Noting from the definitions that
\begin{equation}\label{DLk}
 \overline{W}_n^{\rm G^*}(x_1,\dots,x_n;N,\kappa)  = (\sqrt{2 \kappa})^{-n}   \overline{W}_n^{\rm G}(x_1/\sqrt{2 \kappa},\dots, x_n/\sqrt{2 \kappa};N,\kappa),
  \end{equation}
 it follows from  (\ref{1.4e}) that
 $$
 \overline{W}_n^{\rm G^*}(x_1,\dots,x_n;N,\kappa) = 
 (-i/\sqrt{2} \kappa)^n  \overline{W}_n^{\rm G}(i x_1/\sqrt{2 },\dots,ix_n/\sqrt{2};-\kappa N, \kappa^{-1}).
 $$
 Hence, substituting for $\overline{W}_n^{\rm G^*}$ in (\ref{2.29a}) it
follows
 \begin{align*}
  \overline{W}_n^{0, \rm G^*}(x_1,\dots,x_n;\alpha) & =
{1 \over \alpha} \Big ( {-i \over \sqrt{2}} \Big )^n  \lim_{N \to \infty}  (N/ \alpha)^{n-1}
 \overline{W}_n^{\rm G}(i x_1/ \sqrt{2},\dots,ix_n/\sqrt{2};-\alpha,N/\alpha) \\
& = {1 \over \alpha} \Big ( {-i \over \sqrt{2}} \Big )^n  W_n^{0, \rm G}(ix_1/\sqrt{2},\dots,ix_n/\sqrt{2};-\alpha),
\end{align*}
where the second equality follows from (\ref{2.29m})
Setting $- \alpha = N$ gives (\ref{w1j}).

The equation (\ref{w1k}) now follows by substituting the expansions (\ref{x1}) and (\ref{x2}) in
(\ref{w1j}) and equating coefficients of like powers.
\end{proof}

\begin{remark}
1.~The functional equation (\ref{1.4e}) also has consequence for the expansion coefficients
$\{ \mu^{\rm G}_{p_1,\dots,p_n}(N,\kappa) \}$ in (\ref{1.16a}), by providing a generalisation of
(\ref{w4+}).  Thus we see
\begin{equation}\label{w4+1} 
  \mu_{p_1,\dots,p_n}^{\rm G}(N,\kappa) =  (1/i \sqrt{\kappa})^{2n + \sum_{k=1}^n p_k}
      \mu_{p_1,\dots,p_n}^{\rm G}(-\kappa N,\kappa^{-1} ).
  \end{equation}
  
  \noindent
  2.~A coupled recurrence scheme to compute $\{\mu ^{0, \rm G^*}(p_1,p_2;\alpha) \}$
  has been given in \cite{FM21} (the coupling is with $\{\mu ^{0, \rm G^*}(p;\alpha) \}$).
  In light of (\ref{w1k}) this implies an analogous computation  of $\{\mu ^{0, \rm G}(p_1,p_2;N) \}$.
  
   \noindent
  3.~We can check from (\ref{DLk}) that the functional equation (\ref{1.4e}) is formally unchanged
  by use of the weight $e^{-x^2/2}$ corresponding to the Gaussian $\beta$-ensemble denoted by
  W${}^*$ and thus
  \begin{equation}\label{1.4e+} 
 \overline{W}_n^{\rm G^*}(x_1,\dots,x_n; N, \kappa)
 =  ( - i /   \sqrt{\kappa })^n    \overline{W}_n^{\rm G^*}(i  \sqrt{\kappa} x_1,\dots, i  \sqrt{\kappa}  x_n ; - \kappa N, \kappa^{-1}).
\end{equation}
   \end{remark}

  \section{High-low temperature duality for the Laguerre $\beta$-ensemble}\label{S3} 
   \subsection{The zero temperature limit of $W_1^{\rm L}$}
The considerations of the previous section can be extended to the other classical
   $\beta$-ensembles, namely those with the Laguerre and the Jacobi weights.
   In the Laguerre case, we will denote by "L" the $\beta$-ensemble with weight
   $x^{ \beta a/2} e^{-\beta x/2} \chi_{x > 0}$. We know from \cite{FRW17} that the
   first loop equation is
   \begin{equation}\label{Lc1}
   \Big [ (\kappa - 1) {d \over d x} +
   \Big ( {a \kappa \over x} - \kappa \Big ) \Big ] \overline{W}_1^{\rm L}(x;N,a,\kappa) + {N \kappa \over x} +
   \kappa \Big [ \overline{W}_2^{\rm L}(x,x;N,a,\kappa)) + \overline{W}_1^{\rm L}(x;N,a,\kappa))^2 \Big ] = 0,
   \end{equation}
   where as in the previous section $\kappa = \beta/2$.
 Introducing low temperature expansions of the form
 (\ref{L3}) shows that in the low temperature limit (\ref{Lc1}) reduces  to the Riccati type
 equation for the quantity $ {W}_1^{0,\rm L}(x;N,a)$,
 \begin{equation}\label{Lc2}
  \Big [  {d \over d x} +
   \Big ( {a  \over x} - 1 \Big ) \Big ] {W}_1^{0, \rm L}(x;N,a) + {N \over x} +
   (W_1^{0,\rm L}(x;N,a))^2 = 0,
   \end{equation}
  which can readily be solved in terms of Laguerre polynomials.
  
    \begin{proposition}\label{p3.1}
  We have
   \begin{align}\label{2.14L}
    & W_1^{0, \rm L}(x;N,a)  \nonumber \\ & \quad = {d \over d x} \log L_N^{a-1}(x) = {N \over x} + {d \over d x} \Big ( 1 + N! (N+a-1)!  \sum_{m=1}^\infty {(-x)^{- m} \over m! (N + a -1-  m)! (N - m)!} \Big ) \nonumber \\
    & \quad = {N \over x} + {N (N + a-1) \over x^2}  \bigg (
    1 + {( 2 N + a-2 ) \over x} + { (5 N^2 + (-11+5 a)N +  a ^2 - 5a + a^2) \over  x^2} +  \cdots  \bigg )
   \end{align}
   Furthermore, writing 
   \begin{equation}\label{2.14aL}  
     W_1^{0, \rm L}(x;N,a) = {1 \over x} \sum_{p=0}^\infty {m_{p}^{0, \rm L}(N,a) \over x^{p} }
     \end{equation}
     as is consistent with (\ref{v1}) and (\ref{v2}), we have that $\{  m_{p}^{0, \rm L} \}$ satisfy the recurrence
   \begin{equation}\label{2.15L}  
    m_{k+1}^{0, \rm L}(N,a)  = \sum_{s=0}^{k} m_{s}^{0, \rm L}(N,a)  m_{k-s}^{0, \rm L}(N,a)  - (k + 1-a) m_{k}^{0, \rm L}(N,a), \qquad m_0^{0, \rm L}(N,a) = N,
   \end{equation}
   valid for $k=0,1,\dots$.
   \end{proposition}
   
   \begin{proof}
  Setting
   $$
    W_1^{0, \rm L}(x;N,a) = {d \over d x} \log  f(x), \qquad f(x) \mathop{\sim}\limits_{x \to \infty} c x^N,
   $$
   where $c$ is a nonzero constant, we see from (\ref{Lc2}) that $f$ satisfies the linear second order differential equation
  \begin{equation}\label{2.16L}   
  x f''(x) + (a - x)  f'(x) +  N f(x) = 0.
  \end{equation}
   Up to a proportionality constant, the unique solution of this equation satisfying the required boundary condition is the Laguerre polynomial $L_N^{a-1}(x)$.
  The recurrence (\ref{2.15L}) follows by substituting (\ref{2.14aL}) in (\ref{Lc2}) and equating like powers of $1/x$.
  \end{proof}	
  
  \begin{remark}
  1.~It follows from \cite[Prop.~3.11]{FRW17} (see also \cite{MRW15}) that with Laguerre weight $x^{ \beta a/2} e^{-\beta x/2}$, the corresponding
  moments of the density $m^{(L)}(N,a,\kappa)$ are, for low order, given by
  \begin{align}\label{Lq}
  {1 \over N} m_1^{(L)}(N,a,\kappa) & = N + {1 \over \kappa} (1 - \kappa + \kappa a) \nonumber \\
  {1 \over N} m_2^{(L)}(N,a,\kappa) & = 2 N^2 + {N \over \kappa} (4 - 4 \kappa + 3 a \kappa) +  {1 \over \kappa^2} (2 - 4 \kappa + 2 \kappa^2 + 3 a \kappa - 3 a \kappa^2 + \kappa^2 a^2 ).
  \end{align}
  Expanding with respect to $1/\kappa$, we see the leading terms are the coefficients of $1/x^2$ and $1/x^3$, as is consistent with (\ref{2.14L}). \\
  2.~Writing the PDF for the Laguerre $\beta$-ensemble in Boltzmann factor form
  $$
  e^{-\beta U^{\rm L}}, \qquad U^{\rm L} :=  \sum_{l=1}^N  ( x_l - a  \log x_l) - \sum_{1 \le j < k \le N} \log | x_j - x_k|,
  $$
  from a limiting case of (\ref{BfJ}) below, as considered by Stieltjes in \cite{St85b}, we know that the minimum of $U^{\rm L} $ is unique and
  occurs at the zeros of the Laguerre polynomial $L_N^{a-1}(x)$. Arguing as in the derivation of (\ref{r1}) gives an alternative derivation 
  of the first equality in (\ref{2.14L}).
  \end{remark}
  
   \subsection{High-low temperature duality of $W_1^{\rm L}$}\label{S3.2}
   The recent work \cite{FM21} considered the Laguerre $\beta$-ensemble with weight
   $x^a e^{-x} \chi_{x > 0}$ (to distinguish this from the Laguerre weight of the previous subsection, the
   notation "L${}^*$" will be used) in the scaled high temperature limit obtained by setting $\kappa = \alpha / N$
   and taking $N \to \infty$.
   Writing
    \begin{equation}\label{s3L} 
{1 \over N}  \left\langle\sum_{j=1}^N\frac{1}{x-\lambda_j}\right\rangle^{\rm L^*}
\bigg |_{\kappa =  \alpha / N} = W_1^{0, \rm L^*}(x;\alpha,a) + {1 \over N} W_1^{1, \rm L^*}(x;\alpha,a) + \cdots
  \end{equation}  
  it was shown that $W_1^{0, \rm L^*}$ satisfies the differential equation
   \begin{equation}\label{Lc2+}
  \Big [  - {d \over d x} +
   \Big ( {a  \over x} - 1 \Big ) \Big ] {W}_1^{0, \rm L^*}(x;\alpha,a) + {1 \over x} +
 \alpha  (W_1^{0,\rm L^*}(x;\alpha,a))^2 = 0.
   \end{equation}
   
   We see that upon the substitutions $x=-y$, $W_1^{0, \rm L^*}(x;N,a) = - N g(y)$
   in (\ref{Lc2+}) that $g(y)$ so specifed satisfies precisely the differential equation (\ref{Lc2}), except that
   $a \mapsto -a$ and $\alpha = - N$. 
   Considering the requirement of the behaviour at infinity we conclude
    \begin{equation}\label{WS1}
    - {1 \over N} W_1^{0,\rm L}(-x;N,a) = W_1^{0,\rm L^*}(x;-N,-a).
     \end{equation}
   In fact we know from \cite{ABMV13} (see also \cite{FM21}) that
   \begin{equation}\label{Lc3} 
  W_1^{0, \rm L^*}(x;\alpha,a) = - {1 \over \alpha} {d \over dx} \log \Big ( x^{-a/2} e^{-x/2} W_{-\alpha-a/2,(1+a)/2}(-x) \Big ),
  \end{equation}
  where $W_{\mu,\kappa}(z)$ denotes the Whittaker function. Substituting this in the right hand side of 
  (\ref{WS1}), and substituting the first equality of (\ref{2.14L}) for the left hand side shows that we must have
    \begin{equation}\label{Lc4} 
  L_N^{a-1}(x) = C e^{x/2} x^{a/2} W_{N + a/2, (1 - a)/2}(x)
   \end{equation}
   for some constant $C$. This is a known identity \cite[Eq.~(13.18.17) and the general fact $W_{\mu,\kappa}(z)
   = W_{\mu,-\kappa}(z)$]{DLMF}.
  Also,  writing
  \begin{equation}\label{w2L}
  W_1^{0, \rm L^*}(x;\alpha, a) = {1 \over x} \sum_{p=0}^\infty { m_{p}^{0, \rm L^*}(\alpha,a) \over x^{p}},
   \end{equation}
   it follows from (\ref{2.14aL}) and (\ref{WS1}) that
   \begin{equation}\label{w3L}  
   {(-1)^p \over N } m_{p}^{0, \rm L}(N,a) = m_{p}^{0, \rm L^*}(-N,-a).
   \end{equation}
   Substituting in (\ref{2.15L}), this implies
    \begin{equation}\label{2.15L+}  
    m_{k+1}^{0, \rm L^*}(-N,-a)  = - N \sum_{s=0}^{k} m_{s}^{0, \rm L}(-N,-a)  m_{k-s}^{0, \rm L^*}(-N,-a)  + (k + 1-a) m_{k}^{0, \rm L^*}(-N,-a),
   \end{equation}
subject to the initial condition
   $m_0^{0, \rm L^*}(-N,-a) = 1$. With $(-N,-a)$ replaced by $(\alpha,a)$ this is known from \cite[Eq.~(4.13)]{FM21}.

  \subsection{High-low temperature duality of $W_n^{\rm L}$}\label{S3.3} 
  The equation (\ref{Lc1})  is the case $n=1$ of a hierarchy of loop equations which for
  $n \ge 2$ read \cite[Eq.~(3.9)]{FRW17}
  \begin{align}\label{eq:loop_laguerre}
0&=\left[(\kappa-1)\frac{\partial}{\partial x_1}+\left(\frac{\kappa a}{x_1}-\kappa \right)\right]\overline{W}_n^{\rm L}(x_1,J_n) \nonumber
\\&\quad+\sum_{k=2}^n\frac{\partial}{\partial x_k}\left\{\frac{\overline{W}_{n-1}^{\rm L}(x_1,\ldots,\hat{x}_k,\ldots,x_n)-\overline{W}_{n-1}^{\rm L}(J_n)}{x_1-x_k}+\frac{1}{x_1}\overline{W}_{n-1}^{\rm L}(J_n)\right\}\nonumber
\\&\quad+\kappa\left[\overline{W}_{n+1}^{\rm L}(x_1,x_1,J_n)+\sum_{J\subseteq J_n}\overline{W}_{|J|+1}^{\rm L}(x_1,J)\overline{W}_{n-|J|}^{\rm L}(x_1,J_n\setminus J)\right].
\end{align}
 Moreover, we know from \cite{FRW17} that these equations become triangular upon an appropriate $1/N$ expansion of each
 of the $ \overline{W}_{n}^{\rm L} $, and so uniquely determine the latter.
 Analogous to (\ref{1.16a}), introducing the expansion about infinity
 \begin{equation}\label{1.4fL} 
 \overline{W}_n^{\rm L}(x_1,\dots,x_n; N, \kappa,a) =
 {1 \over x_1 \cdots x_n} \sum_{p_1,\dots,p_n = 0}^\infty {\mu_{p_1,\dots,p_n}^{\rm L}(N,\kappa,a) \over
 x_1^{p_1} \cdots x_n^{p_n}},
 \end{equation}
 we can give a meaning to $ \overline{W}_n^{\rm L}$ for general $N,a,\kappa$, since it is known that the
 $\mu_{p_1,\dots,p_n}^{\rm L^*}$ are polynomials in $N,a,\kappa$; this latter point is illustrated by
 (\ref{Lq}) for $n=1$. With this understood, we can make use of the loop equations to derive a functional equation
 for the $ \overline{W}_n^{\rm L}$.
 
 \begin{proposition}\label{p3.1a}
For the Laguerre $\beta$-ensemble with weight $x^{\beta a/2} e^{-\beta x/2}$,
denote the quantity (\ref{1.4c+}) by $ \overline{W}_n^{\rm L}(x_1,\dots,x_n; N, \kappa,a)$
for $n \ge 2$, and similarly the quantity (\ref{1.4d}) in the case $n=1$.
Extend their meaning to general values of $N$ using (\ref{1.4fL}).
We have
\begin{equation}\label{1.4eL} 
 \overline{W}_n^{\rm L}(x_1,\dots,x_n; N, \kappa,a)
 =      \overline{W}_n^{\rm L}( -\kappa x_1,\dots,   -\kappa x_n ; - \kappa N, \kappa^{-1},- \kappa a).
\end{equation}
Equivalently
\begin{equation}\label{1.4eLp} 
\mu_{p_1,\dots,p_n}^{\rm L}(N,\kappa,a)  = (-\kappa)^{-n-\sum_{l=1}^n p_l} 
\mu_{p_1,\dots,p_n}^{\rm L}(-\kappa N,\kappa^{-1},- \kappa a). 
\end{equation}
\end{proposition}

\begin{proof}
As in the proof of Proposition \ref{p1.1}) the equation (\ref{1.4eL}) follows as an
invariance of the full set of loop equations. The identity (\ref{1.4eLp}) now follows 
upon substituting (\ref{1.4fL}) in (\ref{1.4eL}) and equating coefficients.
\end{proof}

 We can use Proposition \ref{p3.1a} to obtain a generalisation of (\ref{w1}) and (\ref{w1k}).
  For this purpose, analogous to (\ref{2.29m}) expand for $\kappa \to \infty$
  \begin{equation}\label{1.4eLq} 
   \overline{W}_n^{\rm L}(x_1,\dots,x_n; N, \kappa,a) = {1 \over \kappa^{n-1}}   {W}_n^{0,\rm L}(x_1,\dots,x_n; N,a) + \cdots
  \end{equation}
  and analogous to (\ref{2.29a}) also expand for $N \to \infty$
  \begin{equation}\label{1.4eLr}  
  {1 \over N}    \overline{W}_n^{\rm L^*}(x_1,\dots,x_n; N, \kappa,a) \Big |_{\kappa = \alpha/N}  = 
   {W}_n^{0,\rm L^*}(x_1,\dots,x_n; \alpha,a) + \cdots 
   \end{equation}  
  Further expand the terms on the right hand side of each of these,
  \begin{equation}\label{x1L}
 {W}_n^{0, \rm L}(x_1,\dots,x_n;N,a) =
 {1 \over x_1 \cdots x_n} \sum_{p_1,\dots,p_n = 0}^\infty {\mu_{p_1,\dots,p_n}^{0, \rm L}(N,a) \over
 x_1^{p_1} \cdots x_n^{p_n}}.
 \end{equation}
 and
 \begin{equation}\label{x2L}
 {W}_n^{0, \rm L^*}(x_1,\dots,x_n; \alpha,a) =
 {1 \over x_1 \cdots x_n} \sum_{p_1,\dots,p_n = 0}^\infty {\mu_{p_1,\dots,p_n}^{0, \rm L^*}(\alpha,a) \over
 x_1^{p_1} \cdots x_n^{p_n}}
 \end{equation}
(cf.~(\ref{x1}) and (\ref{x2})).

\begin{corollary}
We have
 \begin{equation}\label{Wh}
   - {1 \over N}  {W}_n^{0,\rm L}(x_1,\dots,x_n; N ,a) =   {W}_n^{0,\rm L^*}(-x_1,\dots,-x_n; -N ,-a) 
 \end{equation}
 and
\begin{equation}\label{Wh1}
 {1 \over N}  \mu_{p_1,\dots,p_n}^{0, \rm L}(N,a)  =  (-1)^{p_1+\dots+p_n+n-1} 
\mu_{p_1,\dots,p_n}^{0, \rm L^*}(-N,-a).  
 \end{equation}
 \end{corollary}  

\begin{proof}
Noting from the definitions that
$$
  \overline{W}_n^{\rm L^*}(x_1,\dots,x_n; N, \kappa,a)  = \kappa^{-n}   \overline{W}_n^{\rm L}(x_1/\kappa,\dots,x_n/\kappa; N, \kappa,a/\kappa),
  $$
  it follows from (\ref{1.4eL}) that
  $$
  \overline{W}_n^{\rm L^*}(x_1,\dots,x_n; N, \kappa,a)  = \kappa^{-n}
  \overline{W}_n^{\rm L}(-x_1,\dots,-x_n; - \kappa N, \kappa^{-1},-a).
  $$
 Setting $\kappa = \alpha/N$, dividing both sides by $N$ and
  taking the limit $N \to \infty$ on both sides using (\ref{1.4eLq}) and (\ref{1.4eLr}) shows
  $$
  \overline{W}_n^{\rm L^*}(x_1,\dots,x_n;  \alpha ,a)  =   {1 \over \alpha}  \overline{W}_n^{\rm L}(-x_1,\dots,-x_n; - \alpha, -a).
  $$
  Changing the sign of each of $(x_1,\dots,x_n;  \alpha ,a)$ in this equation, interchanging the role of the LHS and RHS
  and setting $\alpha = N$ gives  \ref{Wh}).
  Substituting (\ref{x1L}) and (\ref{x2L}) gives (\ref{Wh1}).
\end{proof}

 \section{High-low temperature duality for the Jacobi $\beta$-ensemble}\label{S4} 
   \subsection{The zero temperature limit of $W_1^{\rm J}$}
We will 
 denote by "J" the $\beta$-ensemble implied by the weight
   $x^{ \beta a/2}  (1 - x)^{\beta b/2}\chi_{0 < x  < 1}$. 
   Writing the eigenvalue PDF in Boltzmann factor form,
   as proportional  to
    \begin{equation}\label{BfJ}
   e^{-\beta U^{\rm J}}, \qquad U^{\rm J} :=  - \sum_{l=1}^N (a \log x_l + b  \log (1 - x_l ) )-\sum_{1 \le j < k \le N} \log | x_j - x_k|,
    \end{equation}
    we know from the work of Stieltjes \cite{St85b} that the minimum of $U^{\rm J}$ for $x_l \in (0,1)$ is unique and
    occurs at the zeros of the Jacobi polynomials $P_N^{(a-1,b-1)}(1-2x)$.
    As in the Gaussian and Laguerre cases above, this result can be reclaimed using
    a loop equation formalism.

   Thus we know from \cite{FRW17} that for the Jacobi $\beta$-ensemble the
   first loop equation is
  \begin{align}
0&=\left((\kappa-1)\frac{d}{d x_1}+\left(\frac{\kappa a}{x_1}-\frac{\kappa b}{1-x_1}\right)\right)\overline{W}_1^{0,\rm J}(x_1;N,a,b,\kappa)
+ \frac{1}{x_1(1-x_1)}\left[(\kappa a +\kappa b +1)N+\kappa N(N-1)\right] \nonumber
\\&\quad+\kappa\left[\overline{W}_{2}^{0,\rm J}(x_1,x_1;N,a,b,\kappa)+ (\overline{W}_{1}^{0,\rm J}(x_1;N,a,b,\kappa))^2 \right].\label{4.5J}
\end{align}
 Introducing low temperature expansions of the form
 (\ref{L3}) we see that in the low temperature limit (\ref{4.5J}) reduces  to a Riccati type
 equation for the quantity $ W_1^{0,\rm J}(x;N,a,b)$. This equation reads
 \begin{equation}
 0 =\left( \frac{d}{d x}+ \frac{ a}{x}-\frac{ b}{1-x} \right){W}_1^{0,\rm J}(x;N,a,b) 
+\frac{(a+b- 1)N + N^2 }{x(1-x)}
+({W}_{1}^{0,\rm J}(x;N,a,b))^2. \label{4.5a}
\end{equation}
 Solving (\ref{4.5a}) with an appropriate boundary condition  gives rise to Jacobi  polynomials.

    \begin{proposition}\label{p5.1}
  We have
   \begin{align}\label{2.14J}
    & W_1^{0, \rm J}(x;N,a,b) = {d \over d x} \log P_N^{(a-1,b-1)}(1-2x)   \nonumber \\  & \quad 
    = {N \over x} + {d \over d x} \bigg ( 1 + {N! \Gamma(a+N)  \over \Gamma(a+b+2N-1)}
    \sum_{m=1}^\infty {(-x)^{- m}\Gamma(a+b+2N-m-1) \over m! (N - m)! \Gamma(a+N-m)} \bigg ) \nonumber \\
    & \quad = {N \over x} + {N (N + a-1) \over (2N+a+b-2) x^2}   \nonumber \\  & \quad  \qquad + {N (N + a-1)
    (4 + a^2 + a (-4 + b) - 2 b  + (-7 + 3 a + 2 b) N + 3 N^2) \over (-3 + a + b + 2 N) (-2 + a + b + 2 N)^2 x^3} +   \cdots 
   \end{align}
   Furthermore, writing 
   \begin{equation}\label{2.14aJ}  
     W_1^{0, \rm J}(x;N,a,b) = {1 \over x} \sum_{p=0}^\infty {m_{p}^{0, \rm J}(N,a,b) \over x^{p} }
     \end{equation}
   we have that $\{  m_{p}^{0, \rm J} \}$ satisfy the recurrence
   \begin{equation}\label{2.15J}  
  m_{p}^{0, \rm J} = {1 \over 2N + a + b -1-p}  \Big ( (N+a-1 )N + b \sum_{s=1}^{p-1} m_{s}^{0,\rm J}  + N  \sum_{s=1}^{p-1} m_{s,0}^{0,\rm J} m_{p-s}^{0,\rm J} \Big )
    , \quad m_0^{0, \rm J} = N,
   \end{equation}
   valid for $p=1, 2, \dots$.
   \end{proposition}
   
   \begin{proof}
  Let $c$ be a nonzero constant and set
   $$
    W_1^{0, \rm J}(x;N,a,b) = {d \over d x} \log  f(x), \qquad f(x) \mathop{\sim}\limits_{x \to \infty} c x^N.
   $$
 We see from (\ref{4.5a}) that $f$ satisfies the second order linear differential equation
  \begin{equation}\label{2.16J}   
  x(1-x) f''(x) + (a(1-x) - b x)  f'(x) +  ((a+b-1)N + N^2)  f(x) = 0.
  \end{equation}
  Up to a proportionality constant, the unique solution of this equation satisfying the required boundary condition is the Jacobi polynomial $P_N^{(a-1,b-1)}(1-2x)$.
  The recurrence (\ref{2.15J}) follows by substituting (\ref{2.14aJ}) in (\ref{4.5a}) and equating like powers of $1/x$.
  \end{proof}
  
  \subsection{High-low temperature duality of $W_1^{\rm J}$}\label{S4.2}
   The recent work \cite{FM21} considered the Jacobi $\beta$-ensemble with weight
   $x^a  (1 - x)^b\chi_{0 < x  < 1}$ (to distinguish this from the Jacobi weight of the previous subsection, the
   notation "J${}^*$" will be used) in the scaled high temperature limit $\kappa = \alpha / N$ and $N \to \infty$.
   Writing the analogue of the expansion (\ref{s3L}),
  it was shown that $W_1^{0, \rm J^*(x;\alpha,a,b}$ satisfies the differential equation
  \begin{equation}
		\label{eq:Jac_N_order}
		\left(-{d \over d x} +\frac{a}{x} -\frac{b}{1-x}\right)W_1^{0,\rm J^*}(x;\alpha,a,b) + \frac{1}{x(1-x)}\left(1+a+b+\alpha \right) + \alpha \left(W_1^{0,\rm J^*}(x;\alpha,a,b)\right)^2 = 0.
	\end{equation}
 
 Substituting  $W_1^{0,\rm J^*}(x) = -  g(x)/N$
   we see that $g(x)$ satisfies the differential equation (\ref{2.16J}), except that
   $a \mapsto -a, b \mapsto - b$ and $\alpha = - N$. 
   Considering the requirement of the behaviour at infinity we conclude
    \begin{equation}\label{WS1J}
    - {1 \over N} W_1^{0,\rm J^*}(x;N,a,b) = W_1^{0,\rm J^*}(x;-N,-a,-b).
     \end{equation}
 From \cite{FM21}, in terms of the Gauss hypergeometric function,  we have
   \begin{equation}\label{Jc3} 
  W_1^{0, \rm J^*}(x;\alpha,a,b) = - {1 \over \alpha} {d \over dx} \log \Big (
  x^{-\alpha}{}_2F_1\left( \alpha, \alpha + a +1, 2\alpha + a + b + 2; x^{-1} \right) \Big ).
  \end{equation}
  Substituting this along with the first equality in (\ref{2.14J}) shows
    \begin{equation}\label{Jc4} 
  P_N^{(a-1,b-1)}(1 - 2x) = C   x^{N}{}_2F_1\left( -N, - N  - a +1, -2N - a - b + 2; x^{-1} \right)
   \end{equation}
   for some constant $C$, which can readily be checked by equating like powers of $x$.

   \subsection{High-low temperature duality of $W_n^{\rm J}$}\label{S4.3} 
 The hierarchy of loop equations generalising (\ref{4.5}) for
  $n \ge 2$ reads \cite[Eq.~(4.6)]{FRW17}
  \begin{align}
0&=\left((\kappa-1)\frac{\partial}{\partial x_1}+\left(\frac{\kappa a}{x_1}-\frac{\kappa b}{1-x_1}\right)\right)\overline{W}_n^{\rm J}(x_1,J_n) \nonumber
\\&\quad -\frac{1}{x_1(1-x_1)}\sum_{k=2}^nx_k\frac{\partial}{\partial x_k}\overline{W}_{n-1}^{\rm J}(J_n)\nonumber
\\&\quad+ \sum_{k=2}^n\frac{\partial}{\partial x_k}\left\{\frac{\overline{W}_{n-1}^{\rm J}(x_1,\ldots,\hat{x}_k,\ldots,x_n)-\overline{W}_{n-1}^{\rm J}(J_n)}{x_1-x_k}+\frac{1}{x_1}\overline{W}_{n-1}^{\rm J}(J_n)\right\}\nonumber
\\&\quad+\kappa\left[\overline{W}_{n+1}^{\rm J}(x_1,x_1,J_n)+\sum_{J\subseteq J_n}\overline{W}_{|J|+1}^{\rm J}(x_1,J)\overline{W}_{n-|J|}^{\rm J}(x_1,J_n\setminus J)\right].\label{4.5}
\end{align}
These equations become triangular upon an appropriate $1/N$ expansion of each
 of the $ \overline{W}_{n}^{\rm J} $, and so uniquely determine the latter \cite{FRW17}.
 As in the Gaussian and Laguerre cases, introducing the expansion about infinity
 \begin{equation}\label{1.4fJ} 
 \overline{W}_n^{\rm J}(x_1,\dots,x_n; N, \kappa,a,b) =
 {1 \over x_1 \cdots x_n} \sum_{p_1,\dots,p_n = 0}^\infty {\mu_{p_1,\dots,p_n}^{\rm J}(N,\kappa,a,b) \over
 x_1^{p_1} \cdots x_n^{p_n}},
 \end{equation}
 allows for a meaning to $ \overline{W}_n^{\rm J}$ for general $N,a,b,\kappa$, since it is known that the
 $\mu_{p_1,\dots,p_n}^{\rm J}$ are rational functions in $N,a,b,\kappa$; see \cite{FRW17}.

 \begin{proposition}\label{p4.1}
For the Jacobi $\beta$-ensemble with weight $x^{\beta a/2} (1 - x)^{\beta b /2}$,
denote the quantity (\ref{1.4c+}) by $ \overline{W}_n^{\rm J}(x_1,\dots,x_n; N, \kappa,a,b)$
for $n \ge 2$, and similarly the quantity (\ref{1.4d}) in the case $n=1$.
Extend their meaning to general values of $N$ using (\ref{1.4fJ}).
We have
\begin{equation}\label{1.4eJ} 
 \overline{W}_n^{\rm J}(x_1,\dots,x_n; N, \kappa,a,b)
 =   (- \kappa)^{-n}   \overline{W}_n^{\rm J}(  x_1,\dots,   x_n ; - \kappa N, \kappa^{-1},-\kappa a,-\kappa b).
\end{equation}
Equivalently
\begin{equation}\label{1.4eJp} 
\mu_{p_1,\dots,p_n}^{\rm J}(N,\kappa,a,b)  = (-\kappa)^{-n} 
\mu_{p_1,\dots,p_n}^{\rm J}(-\kappa N,\kappa^{-1},-a,-b). 
\end{equation}
\end{proposition}

\begin{proof}
Proceeding as in the Gaussian and Laguerre cases,
 the functional equation (\ref{1.4eJ}) follows as an
invariance of the full set of loop equations, and (\ref{1.4eJp}) is
a consequence of this which follows
upon substituting (\ref{1.4fJ}) in (\ref{1.4eJ}) and equating coefficients.
\end{proof}

 As in the Gaussian and Laguerre cases,
 we can use Proposition \ref{p4.1} to obtain a generalisation of (\ref{w1}) and (\ref{w1k}).
To begin we expand for $\kappa \to \infty$
  \begin{equation}\label{1.4eJq} 
   \overline{W}_n^{\rm J}(x_1,\dots,x_n; N, \kappa,a,b) = {1 \over \kappa^{n-1}}   {W}_n^{0,\rm J}(x_1,\dots,x_n; N,a,b) + \cdots
  \end{equation}
  and also expand for $N \to \infty$
  \begin{equation}\label{1.4eJr}  
  {1 \over N}    \overline{W}_n^{\rm J^*}(x_1,\dots,x_n; N, \kappa,a,b) \Big |_{\kappa = \alpha/N}  = 
   {W}_n^{0,\rm J}(x_1,\dots,x_n; \alpha,a,b) + \cdots 
   \end{equation}  
  Then we expand the terms on the right hand side of each of these,
  \begin{equation}\label{x1J}
 {W}_n^{0, \rm J}(x_1,\dots,x_n;N,a,b) =
 {1 \over x_1 \cdots x_n} \sum_{p_1,\dots,p_n = 0}^\infty {\mu_{p_1,\dots,p_n}^{0, \rm J^*}(N,a,b) \over
 x_1^{p_1} \cdots x_n^{p_n}},
 \end{equation}
 and
 \begin{equation}\label{x2J}
 {W}_n^{0, \rm J^*}(x_1,\dots,x_n; \alpha,a) =
 {1 \over x_1 \cdots x_n} \sum_{p_1,\dots,p_n = 0}^\infty {\mu_{p_1,\dots,p_n}^{0, \rm J^*}(\alpha,a) \over
 x_1^{p_1} \cdots x_n^{p_n}}.
 \end{equation}

\begin{corollary}
We have
 \begin{equation}\label{WhJ}
  {W}_n^{\rm J^*}(x_1,\dots,x_n; N,\kappa ,a,b)  = {(-1)^n \over N}   {W}_n^{\rm J}(x_1,\dots,x_n; -N ,-a,-b)
 \end{equation}
 and
\begin{equation}\label{Wh1J}
\mu_{p_1,\dots,p_n}^{0, \rm J^*}(N,a,b)  =   {(-1)^n \over N} 
\mu_{p_1,\dots,p_n}^{0, \rm J}(-N,-a,-b)  
 \end{equation}
 \end{corollary}  

\begin{proof}
It follows from (\ref{1.4eJ}) and the definition of J${}^*$ that
$$
{1 \over N}  \overline{W}_n^{\rm J^*}(x_1,\dots,x_n; N, \kappa,a,b) =   {(-\kappa)^{-n}\over N}  {W}_n^{\rm J}(x_1,\dots, x_n; -\kappa N ,\kappa^{-1},-a,-b).
 $$
 Taking the limit $N \to \infty$ on both sides using (\ref{1.4eJq}) and (\ref{1.4eJr}) implies   (\ref{WhJ}).
  Substituting (\ref{x1J}) and (\ref{x2J}) gives (\ref{Wh1J}).
\end{proof}

\begin{remark}
It is known that a limiting case of the Jacobi $\beta$-ensemble implies the
$\beta$-ensemble on the circle corresponding to (\ref{0.5});
see \cite[\S 3.9]{Fo10}. In the case of the latter, playing the role  of the statistics
(\ref{1.4c+}) and (\ref{1.4d}) is the modification of the $A_i$ therein by
$$
A_i = \sum_{p=1}^N {e^{i \theta_p} + z_i \over e^{i \theta_p} - z_i }.
$$
From \cite[Prop.~4.7]{WF15} we know
\begin{equation}\label{Cz}
\overline{W}_n(z_1,\dots,z_n;N,\kappa) = (- \kappa)^{-n}  \overline{W}_n(z_1,\dots,z_n;-\kappa N,\kappa^{-1})
\end{equation}
(cf.~(\ref{1.4eJ})).
\end{remark}

\subsection*{Acknowledgments}
	The research of PJF is part of the program of study supported
	by the Australian Research Council Centre of Excellence ACEMS,
	and the Discovery Project grant DP210102887.
	An input to this work was the knowledge gained from the presentation of
	M.~Voit as part of the Bielefeld-Melbourne on-line random matrix seminar in late October 2020.
	Thanks are due to  G.~Akemann for organising this. The feedback of the referee, by
	way of a thorough reading and helpful remarks, is much appreciated.

 \providecommand{\bysame}{\leavevmode\hbox to3em{\hrulefill}\thinspace}
\providecommand{\MR}{\relax\ifhmode\unskip\space\fi MR }
\providecommand{\MRhref}[2]{%
  \href{http://www.ams.org/mathscinet-getitem?mr=#1}{#2}
}
\providecommand{\href}[2]{#2}

 \end{document}